\RequirePackage{etoolbox}
\newtoggle{lipics}
\newcommand{\lipicsonly}[1]{\iftoggle{lipics}{#1}{}}
\newcommand{\icfponly}[1]{\iftoggle{lipics}{}{#1}}

\toggletrue{lipics}
\documentclass[a4paper,UKenglish,cleveref]{lipics-v2021}

\hideLIPIcs

\author{David Knothe}
{FZI Research Center for Information Technology, Karlsruhe, Germany}
{knothe@fzi.de}
{https://orcid.org/0009-0007-6700-7337}
{}

\author{Oliver Bringmann}
{University of Tübingen, Tübingen, Germany \and FZI Research Center for Information Technology, Karlsruhe, Germany}
{bringmann@fzi.de}
{https://orcid.org/0000-0002-1615-507X}
{}

\authorrunning{D. Knothe and O. Bringmann}

\Copyright{David Knothe and Oliver Bringmann}



\nolinenumbers

\supplementdetails{Software}{https://github.com/knothed/CompCert-loop}

\EventEditors{John Q. Open and Joan R. Access}
\EventNoEds{2}
\EventLongTitle{42nd Conference on Very Important Topics (CVIT 2016)}
\EventShortTitle{CVIT 2016}
\EventAcronym{CVIT}
\EventYear{2016}
\EventDate{December 24--27, 2016}
\EventLocation{Little Whinging, United Kingdom}
\EventLogo{}
\SeriesVolume{42}
\ArticleNo{23}

\title{Combining Small-Step and Big-Step Semantics to Verify Loop Optimizations}

\usepackage[utf8]{inputenc}
\usepackage[T1]{fontenc}
\usepackage{amsmath}

\usepackage{amssymb}
\usepackage{amsthm}
\usepackage{etoolbox}
\usepackage{subcaption}
\usepackage{mathtools}
\usepackage{thmtools}
\usepackage{stmaryrd}
\usepackage{tikz}
\usepackage{makecell}
\usetikzlibrary{arrows.meta}
\usetikzlibrary{decorations.shapes}
\usetikzlibrary{tikzmark}
\usetikzlibrary{math}
\usetikzlibrary{cd}
\usetikzlibrary{nfold}
\usepackage{graphicx}
\usepackage{xcolor}
\usepackage{xspace}
\usepackage{framed}
\usepackage{bussproofs}
\usepackage{hyperref}
\usepackage{mdframed}
\usepackage{float}
\usepackage{afterpage}
\usepackage[frozencache,cachedir=.]{minted}

\let\oldqedsymbol\qedsymbol
\newcommand{\coqed}{\hfill${\small\textnormal{[Rocq] \oldqedsymbol}}$}
\newcommand{\localcoqed}{\renewcommand\qedsymbol{\small\textnormal{[Rocq] \oldqedsymbol}}}

\makeatletter
\newenvironment{centeredmath}{%
	\@fleqnfalse 
}{%
	\ignorespacesafterend
}

\makeatletter

\pgfarrowsdeclare{tonew}{tonew}
{
	\ifdim\pgfgetarrowoptions{tonew}=-1pt%
	\pgfutil@tempdima=0.84pt%
	\advance\pgfutil@tempdima by1.3\pgflinewidth%
	\pgfutil@tempdimb=0.21pt%
	\advance\pgfutil@tempdimb by.625\pgflinewidth%
	\else%
	\pgfutil@tempdima=\pgfgetarrowoptions{tonew}%
	\pgfarrowsleftextend{+-0.8\pgfutil@tempdima}%
	\pgfarrowsrightextend{+0.2\pgfutil@tempdima}%
	\fi%
}
{
	\ifdim\pgfgetarrowoptions{tonew}=-1pt%
	\pgfutil@tempdima=0.28pt%
	\advance\pgfutil@tempdima by.3\pgflinewidth%
	\pgfutil@tempdimb=0pt,%
	\else%
	\pgfutil@tempdima=\pgfgetarrowoptions{tonew}%
	\multiply\pgfutil@tempdima by 100%
	\divide\pgfutil@tempdima by 375%
	\pgfutil@tempdimb=0.4\pgflinewidth%
	\fi%
	\pgfsetdash{}{+0pt}
	\pgfsetroundcap
	\pgfsetroundjoin
	\pgfpathmoveto{\pgfpointorigin}
	\pgflineto{\pgfpointadd{\pgfpoint{0.75\pgfutil@tempdima}{0bp}}
		{\pgfqpoint{-2\pgfutil@tempdimb}{0bp}}}
	\pgfusepathqstroke
	\pgfsetlinewidth{0.8\pgflinewidth}
	\pgfpathmoveto{\pgfpointadd{\pgfqpoint{-3\pgfutil@tempdima}{4\pgfutil@tempdima}}
		{\pgfqpoint{\pgfutil@tempdimb}{0bp}}}
	\pgfpathcurveto
	{\pgfpointadd{\pgfqpoint{-2.75\pgfutil@tempdima}{2.5\pgfutil@tempdima}}
		{\pgfqpoint{0.5\pgfutil@tempdimb}{0bp}}}
	{\pgfpointadd{\pgfqpoint{0pt}{0.25\pgfutil@tempdima}}
		{\pgfqpoint{-0.5\pgfutil@tempdimb}{0bp}}}
	{\pgfpointadd{\pgfqpoint{0.75\pgfutil@tempdima}{0pt}}
		{\pgfqpoint{-\pgfutil@tempdimb}{0bp}}}
	\pgfpathcurveto
	{\pgfpointadd{\pgfqpoint{0pt}{-0.25\pgfutil@tempdima}}
		{\pgfqpoint{-0.5\pgfutil@tempdimb}{0bp}}}
	{\pgfpointadd{\pgfqpoint{-2.75\pgfutil@tempdima}{-2.5\pgfutil@tempdima}}
		{\pgfqpoint{0.5\pgfutil@tempdimb}{0bp}}}
	{\pgfpointadd{\pgfqpoint{-3\pgfutil@tempdima}{-4\pgfutil@tempdima}}
		{\pgfqpoint{\pgfutil@tempdimb}{0bp}}}
	\pgfusepathqstroke
}

\makeatother

\pgfsetarrowoptions{tonew}{-1pt}
\pgfkeys{/tikz/.cd, arrowhead/.default=-1pt, arrowhead/.code={
		\pgfsetarrowoptions{tonew}{#1},
}}

\tikzset{decorate sep/.style 2 args=
	{decorate,decoration={shape backgrounds,shape=circle,shape size=#1,shape sep=#2}}}

\newcommand{\dotDownarrow}{%
	\mathrel{%
		\hspace{0.2em}%
		\tikz[baseline={(0,-0ex)}]{
			\draw[decorate sep={0.07mm}{0.57mm},fill]
			(-0.2ex,1.5ex) -- (-0.2ex,0.3ex);%
			\draw[decorate sep={0.07mm}{0.57mm},fill]
			(0.2ex,1.5ex) -- (0.2ex,0.3ex);%
			\begin{scope}[transform canvas={yscale=1.4}]
				\draw[-tonew, arrowhead=2.5pt]
				(0,-0.249ex) -- (0,-0.25ex);%
			\end{scope}
		}%
		\hspace{0.2em}%
	}%
}

\newcommand{\BSEval}[1]{\Downarrow^{#1}}
\newcommand{\BSEvalinf}[1]{\Downarrow^{#1}_\infty}
\newcommand{\BSPartial}[1]{\dotDownarrow^{#1}}
\newcommand{\BSEvalinfx}[2]{\prescript{\scriptscriptstyle (\!#2\!)\hspace{-.2em}}{}{\Downarrow^{#1}_\infty}}

\newcommand{\SSStep}[1]{\overset{#1}{\rightarrow}}
\newcommand{\SSSteptext}[1]{{\rightarrow}^{#1}}
\newcommand{\SSStar}[1]{{\overset{#1}{\longrightarrow}}\!^{\scriptstyle *}}
\newcommand{\SSForever}[1]{{\overset{#1}{\longrightarrow}}\!_{\scriptstyle \infty}}
\newcommand{\SSForeverx}[2]{{\overunderset{#1}{\smash{\raisebox{0.4em}{$\scriptscriptstyle (\!#2\!)$}}}{\longrightarrow}}\!_{\scriptstyle \infty}}

\newcommand{\forever}[1]{{\text{div}}^{#1}}
\newcommand{\foreversilent}{{\text{div\_silent}}^{\emptyTrace}}
\newcommand{\foreverfinite}[1]{{\text{div\_fin}}^{#1}}
\newcommand{\reacts}[1]{{\text{reacts}}^{#1}}

\newcommand{\emptyTrace}{\varepsilon}
\newcommand{\traceConcat}{\cdot}
\newcommand{\prefixStrict}{\! < \!}
\newcommand{\prefix}{\! \leq \!}
\newcommand{\neqTrace}{\! \neq \!}
\newcommand{\refines}{\preccurlyeq}
\newcommand{\stateAbstr}{\sigma}

\DeclareMathOperator{\Beh}{Beh}
\DeclareMathOperator{\behOk}{\texttt{OK}}
\DeclareMathOperator{\behDiv}{\infty}
\DeclareMathOperator{\behErr}{\texttt{err}}

\newcommand{\fwpreserv}{\rightrightarrows}
\newcommand{\bwpreserv}{\leftleftarrows}

\newcommand{\Dfrac}[2]{%
	\ooalign{%
		$\genfrac{}{}{2.2pt}0{\hphantom{#1}}{\hphantom{#2}}$\cr%
		$\color{white}\genfrac{}{}{1pt}0{\color{black}{#1}}{\color{black}{#2}}$}%
}
\newcommand{\Sfrac}[2]{%
	\ooalign{%
		$\color{black}\genfrac{}{}{0.6pt}0{#1}{#2}$
	}
}

\newcommand{\rulefignoline}[2]{
	\vspace{2pt}
	\begin{minipage}{#1\textwidth}
		\begin{centeredmath}
			\[ #2 \]
		\end{centeredmath}
	\end{minipage}
	\vspace{2pt}
}

\newcommand{\rulefig}[3]{
	\vspace{2pt}
	\begin{minipage}{#1\textwidth}
		\centering
		\begin{centeredmath}
			\[ \Sfrac{#2}{#3} \]
		\end{centeredmath}
	\end{minipage}
	\vspace{2pt}
}

\newcommand{\rulefigannot}[4]{
	\vspace{2pt}
	\begin{minipage}{#1\textwidth}
		\centering
		\begin{centeredmath}
			\[ [#2] \Sfrac{#3}{#4} \]
		\end{centeredmath}
	\end{minipage}
	\vspace{2pt}
}

\newcommand{\rulefigdouble}[3]{
	\vspace{2pt}
	\begin{minipage}{#1\textwidth}
		\centering 
		\begin{centeredmath}
			\[ \Dfrac{#2}{#3} \]
		\end{centeredmath}
	\end{minipage}
	\vspace{2pt}
}

\newcommand{\rulefigdoubleannot}[4]{
	\vspace{2pt}
	\begin{minipage}{#1\textwidth}
		\centering 
		\begin{centeredmath}
			\[ [#2] \Dfrac{#3}{#4} \]
		\end{centeredmath}
	\end{minipage}
	\vspace{2pt}
}

\newcommand{\indep}[2]{\textit{indep}\ #1\ #2}
\DeclareMathOperator{\rep}{repeat}
\newcommand{\BSLoop}[2]{\mathop{\prescript{\LOOP}{\scriptstyle #1}{\Downarrow}^{#2}}}
\newcommand{\BSLoopinf}[2]{\mathop{\prescript{\LOOP}{\scriptstyle #1}{\Downarrow}^{#2}_\infty}}

\newcommand{\Cminor}{\texttt{Cminor}\xspace}
\newcommand{\eval}[2]{\text{eval}\ #1\ #2}
\newcommand{\extcall}[2]{\text{external\_call}\ #1\ #2}
\newcommand{\istrue}[1]{\text{is\_true}\ #1}
\newcommand{\isfalse}[1]{\text{is\_false}\ #1}

\newcommand{\kseq}{\boldsymbol{\cdot}}
\DeclareMathOperator{\SKIP}{\mathsf{nop}}
\newcommand{\STORE}[2]{#1 \leftarrow #2}
\newcommand{\IFTHENELSE}[3]{#1\ ?\ #2 : #3}
\newcommand{\SEQ}[2]{#1 \kseq #2}
\DeclareMathOperator{\LOOP}{\circlearrowright}
\newcommand{\BLOCK}[1]{\llbracket {#1} \rrbracket}
\DeclareMathOperator{\EXIT}{\mathsf{exit}}
\DeclareMathOperator{\EXTCALL}{\mathsf{unsafe}}

\DeclareMathOperator{\outOk}{\texttt{OK}}
\newcommand{\outExit}[2]{\texttt{Exit}\ #1\ #2}
\DeclareMathOperator{\outPartial}{\texttt{Partial}}

\newcommand{\Kseq}[2]{#1 \kseq #2}
\newcommand{\Kstop}{\mathsf{stop}}
\newcommand{\Kblock}[1]{\llbracket #1 \rrbracket}

\newcommand{\traceCALL}[2]{\texttt{CALL #1 #2}}

\newcommand{\BSState}[2]{\langle #1,\ #2 \rangle}
\newcommand{\SSState}[3]{\langle #1,\ #2,\ #3 \rangle}

\newcommand{\xdownarrow}[1]{\left\downarrow\vbox to #1{}\right.\kern-\nulldelimiterspace}

\newcommand{\coord}[1]{
	\tikz[remember picture]{
		\coordinate (#1) at (0,0);
	}
}
	
\begin{document}

	\lipicsonly{\maketitle}
	
	\begin{abstract}
Verified compilers aim to guarantee that compilation preserves the observable behavior of source programs.
While small-step semantics are widely used in such compilers, they are not always the most convenient framework for structural transformations such as loop optimizations.
This paper proposes an approach that leverages both small-step and big-step semantics: small-step semantics are used for local transformations, while big-step semantics are employed for structural transformations.

An abstract behavioral semantics is introduced as a common interface between the two styles.
Coinductive big-step semantics is extended to correctly handle divergence with both finite and infinite traces, bringing it on par with the expressiveness of small-step semantics.
This enables the insertion of big-step transformations into the middle of an existing small-step pipeline, thereby fully preserving all top-level semantic preservation theorems.
This approach is practically demonstrated in CompCert by implementing and verifying a few new loop optimizations in big-step Cminor, including loop unswitching and, notably, full loop unrolling.

	\end{abstract}

	\icfponly{\maketitle}

	\section{Introduction}

\emph{Can you trust your compiler?}
This prominent question served as the opening of several influential papers by Leroy and others that introduced CompCert, the first realistic and commercially available verified compiler \cite{compcert-main, compcert-2, compcert-3, compcert-4}.
The objective and value of a verified compiler are clear: through a machine-checked proof of all its transformations and optimizations, many of which are quite complex, bugs in these components are systematically ruled out.
This effectively closes the gap between the correct functionality of a program's written source code and that of the resulting compiled machine code.
While such an approach may be overkill for most ordinary applications, it can be of crucial importance for safety-critical software.

A desirable correctness statement could look like follows:
\emph{If a source program $S$ is compiled into an executable program $C$, then the observable behaviors of $S$ are the same as those of $C$.}

In order to reason about program behavior at all, both the source and the target language must be equipped with formal semantics that specify an execution model.
Two kinds of semantics are commonly used for this purpose: operational (small-step) semantics, which describe the (often hardware-close) step-by-step execution of a program, and natural (big-step) semantics, which adopt a more conceptual, high-level view of execution.
When reasoning about the correctness of a transformation, these two kinds of semantics are therefore naturally suited to different purposes.
Small-step semantics are particularly well suited for transformations that modify code only locally (even if a global analysis is required), such as instruction selection or constant propagation.
Big-step semantics, by contrast, are more appropriate for transformations that operate on the structure of the code, such as inlining or, prominently, loop optimizations.
In this context, Bertot et~al.\ wrote in 2004 \cite{compcert-2}:
\begin{quote}
	 ``\dots we envision to perform some optimizations such as loop optimizations not on the unstructured RTL intermediate code, but on the structured Cminor source code, whose big-step semantics makes it easier to reorder computations without worrying about intermediate computational states that do not match.''
\end{quote}

However, the big-step semantics originally used in CompCert have been completely replaced by small-step semantics in the meantime, with the justification that they are more expressive and provide a common framework for preservation proofs \cite{cc-transition-1, cc-transition-2}.
Consequently, the official CompCert development does not include any loop transformations.
In recent years, methods have been proposed to perform loop optimizations on small-step semantics \cite{cfg-morphisms-1, cfg-morphisms-2, cfg-morphisms-3}.
While these approaches are powerful and promising, they are still complicated due to the reliance on an interpreter and they only work for a class of transformations that are morphisms on the control-flow graph.

The goal of this paper is to demonstrate that big-step semantics are not only better suited for some types of structural transformations than small-step semantics, but can also be integrated into an existing small-step semantics pipeline, while preserving all statements about program behavior.
Using this approach, we verify several new loop transformations in CompCert.
Our methodology and contributions can be summarized as follows:
\begin{enumerate}
	\item We abstract program behavior from both small-step and big-step semantics, creating a minimal common interface between the two. We define properties and prove preservation theorems for this so-called behavioral semantics (\cref{section-2}).
	\item We extend the coinductive big-step semantics of Leroy \cite{coind-bs-leroy} by allowing both finite and infinite traces in the divergence judgment, making it as expressive as that of small-step semantics (\cref{section-coinductive}).
	\item We implement several loop transformations in big-step \Cminor \cite{cminor}, including loop unswitching and full loop unrolling (\cref{section-loop-trafos}). To our knowledge, these transformations have not been verified before in CompCert. We discuss the strengths and limitations of the big-step approach (\cref{section-discussion}).
	\item We prove behavioral equivalence between big-step and small-step \Cminor (\cref{section-cminor}) and integrate all of the above into CompCert (\cref{section-pipeline-integration}). We thereby preserve all of its top-level correctness theorems in their full strength. The development is publicly available \cite{compcert-loop}.
\end{enumerate}
Our contribution is not limited to CompCert: we want to advocate the idea of \textbf{proving certain transformations using small-step semantics and others using big-step semantics} for verified compilers in general.
The behavioral abstractions we employ, the extension of Leroy's coinductive big-step semantics and the thereby enabled full-strength reasoning about big-step behavior can be applied beyond CompCert.
Nonetheless, we use CompCert here as a concrete and illustrative example of our approach.
	\section{Program Behavior}
\label{section-2}

The overarching goal of a verified compiler is to make guarantees about the behavior of a compiled program w.~r.~t. the behavior of the original program.
Thereby, the notion of program behavior is intended to capture two essential aspects: first, whether the program terminates, diverges or goes wrong with undefined behavior; and second, the observable \emph{trace} that is produced during execution.
The trace is a (possibly infinite) sequence of externally visible \emph{events} like calls to \texttt{read} or \texttt{printf}.
It does not capture every single performed operation as this would be too restrictive to allow compiler optimizations such as reordering of operations or changing control flow.
A sensible convention is therefore that the trace captures all calls to external functions as these may be externally visible and no guarantees can be made about their behavior in general.

Let $p$ be a program.%
\footnote{
	We always implicitly include the initial state of the execution environment such as memory and registers.
}
We use the three predicates $\Beh_p(t, \behOk)$, $\Beh_p(\tau, \behDiv)$ and $\Beh_p(t, \behErr)$ to express that $p$ terminates, diverges or goes wrong, respectively.
Note that while termination and going wrong both admit a finite trace $t$, divergence gives rise to a \emph{general} trace $\tau$, that is, either a finite or an infinite trace.
To see why, consider a diverging \texttt{while true}-loop. It may never emit any event or it may call external functions infinitely often (i.~e., more often than any finite number of times).

In general, a program may have multiple behaviors, as external calls can be non-deterministic. This is described in more detail in \cref{section-determinate-receptive} using the \emph{determinacy} property.
To be able to talk about behavior at all, we must equip a program (or an intermediate representation in the compiler) with a semantics. Two prevalent types of semantics for this purpose are small-step (or operational) and big-step (or natural) semantics.

A small-step semantics is characterized by its execution steps.
The judgment $\sigma_1 \SSStep t \sigma_2$ denotes a step from program state $\sigma_1$ to $\sigma_2$, which normally corresponds to the execution of a single operation.
The step emits a trace $t$, which may be the empty trace $\emptyTrace$.
A program $p$ terminates if its execution eventually reaches a final state.
If execution continues indefinitely, $p$ diverges.
If there is no possible step from a non-final state ($\stateAbstr \not \SSStep{}$), $p$ has gone wrong.

Instead of step-by-step execution, a big-step semantics defines termination and divergence in a recursive fashion.
For example, the statement $s_1 \cdot s_2$ terminates if first $s_1$ and then $s_2$ execute successfully.
As described by Leroy \cite{coind-bs-leroy}, a big-step semantics can be defined with an inductive predicate for termination and a coinductive predicate for divergence, which is also the route that CompCert takes.%
\footnote{
	Due to their different natures, these two predicates cannot easily be combined into a single one in Rocq.
	There are a few approaches to defining a unified big-step semantics, such as \cite{pretty-bs} and \cite{pretty-bs-2}, which we will not address here.
}
Typically, big-step semantics lack a going-wrong judgment.

\begin{note}
	We will write $t$ for a finite trace, $\omega$ for an infinite trace and $\tau$ for a (general) trace.
\end{note}
	\subsection{Behavioral Semantics}
\label{section-beh-sem}

We now define \emph{behavioral semantics}. Here, going-wrong is implicitly expressed via partial executions.
Later, we show how both small-step and big-step semantics can be interpreted as behavioral semantics.

Given a program $p$, the behavioral semantics of $p$ is defined by the three predicates ${p \BSEval t \stateAbstr}$ (termimation with final state $\stateAbstr$), $p \BSEvalinf \tau$ (divergence) and $p \BSPartial t$ ($p$ has a partial evaluation with finite trace $t$).
Partial evaluation is required to satisfy the following three properties:
\iftoggle{lipics}{\begin{bracketenumerate}}{\begin{enumerate}}
	\item $p \BSPartial \varepsilon$ holds.
	\item Let $\tau$ be a trace such that either $p \BSEval \tau \stateAbstr$, $p \BSEvalinf \tau$, or $p \BSPartial \tau$ holds. Then, for every finite prefix $t \prefix \tau$, we have $p \BSPartial t$.
	\item Compactness: Let $\omega$ be an infinite trace and assume that $\forall t \prefixStrict \omega: p \BSPartial t$. Then $p \BSEvalinf \omega$ holds.
\iftoggle{lipics}{\end{bracketenumerate}}{\end{enumerate}}

\begin{definition}
\begin{mdframed}[hidealllines=true,nobreak=true,leftmargin=0pt,innerleftmargin=0pt,innertopmargin=0pt]
\noindent The behavior $\Beh_p$ of $p$ is then defined as follows:
\begin{itemize}
	\item $\Beh_p(t, \behOk \stateAbstr) :\Leftrightarrow p \BSEval t \stateAbstr$
	\item $\Beh_p(\tau, \behDiv) :\Leftrightarrow p \BSEvalinf \tau$
	\item $\Beh_p(t, \behErr) :\Leftrightarrow p \BSPartial t \text{and } \neg
	\big(
	p \BSEval t \stateAbstr\ \vee\ p \BSEvalinf t\ \vee\ \ \exists t' \neq \emptyTrace,\ p \BSPartial {t \cdot t'}\! \big)$.
\end{itemize}
\end{mdframed}
\end{definition}

An important implication is that every partial execution gives rise to a behavior. Consequently, by above property (1), every program has at least one behavior:
\begin{lemma}
	If $p \BSPartial t$ then there is a behavior $b$ and a trace $\tau$ such that $\Beh_p(t \traceConcat \tau, b)$.
	\label[lemma]{lemma-partial-has-behavior}
\end{lemma}
\begin{proof}
	Classically%
	\footnote{
		The fundamental \cref{lemma-partial-has-behavior} and therefore many of our theorems are proven classically rather than constructively.
		However, the same holds in CompCert: some key theorems about behaviors, e.~g.~that small-step simulation implies behavior preservation, are proven classically.
	}
	and using property (2) one of the following holds:
	\begin{enumerate}
		\item $\exists \omega$ infinite such that $\forall t' \prefixStrict \omega, p \BSPartial {t \traceConcat t'}\!$.
		By compactness (3) it follows that $p \BSPartial {t \traceConcat \omega}$.
		\item $\exists t'$ finite such that $p \BSPartial {t \traceConcat t'} \wedge\ \forall \tilde t \neqTrace \emptyTrace, \neg (p \BSPartial {t \traceConcat t' \traceConcat \tilde t})$.
		If $p \BSEval {t \traceConcat t'} \stateAbstr$ or $p \BSEvalinf {t \traceConcat t'}$ holds, we are done.
		Otherwise, it follows $\Beh_p(t \traceConcat t', \behErr)$ by definition. \qedhere
	\end{enumerate}
\end{proof}
	\subsection{Properties and Preservation}
\label{section-beh-sem-properties}

\iftoggle{lipics}{\paragraph*{Determinacy and Receptiveness}}{\subsubsection{Determinacy and Receptiveness}}
\label{section-determinate-receptive}

A typical source of nondeterminism is the ability to invoke external functions.
The return value of, for example, a \texttt{read} call is not known in advance.
A semantics is called \emph{determinate} if external calls are the only source of nondeterminism.
In a small-step semantics $\SSStep{}$, determinacy is defined as follows \cite{compcert-tso}:
if there are two possible steps $s \SSSteptext {t_1} s_1$ and $s \SSSteptext {t_2} s_2$, then $t_1 \asymp t_2$ and $t_1 = t_2 \Rightarrow s_1 = s_2$.
Here, ${t_1 \asymp t_2}$ means that the two traces \emph{match}: $|t_1| = |t_2|$ and,  if they represent an external call, then the call arguments coincide, while the return values may differ.%
\footnote{We additionally require that a single step can produce at most one event, i.~e., $s \SSSteptext t s' \Rightarrow |t| \leq 1$.}
Consequently, in a determinate semantics, multiple different executions can arise only after steps that emit a trace event.

Receptiveness is a complementary property:
$s \SSSteptext {t_1} s_1 \wedge t_1 \asymp t_2 \Rightarrow \exists s_2, s \SSSteptext {t_2} s_2$.
CompCert assumes determinacy and receptiveness for all external functions.
We now define analogous properties for a behavioral semantics:

\begin{definition}[Determinacy and Receptiveness] A behavioral semantics on $p$ is
	\iftoggle{lipics}{\begin{alphaenumerate}}{\begin{itemize}}
		\item \emph{determinate}, when:
		\begin{enumerate}
			\item given $|e_{1,2}|=1,\ p \BSPartial {t \traceConcat e_1} \wedge\ p \BSPartial {t \traceConcat e_2}\ \Rightarrow e_1 \asymp e_2$
			\item $\Beh_p(t, b_1) \wedge \Beh_p(t \traceConcat \tau, b_2) \Rightarrow \tau = \emptyTrace \wedge b_1 = b_2$.
		\end{enumerate}
		\item \emph{receptive}, when:
		given $|e_{1,2}|=1,\ p \BSPartial {t \traceConcat e_1} \wedge\ e_1 \asymp e_2 \Rightarrow p \BSPartial{t \traceConcat e_2}$
	\iftoggle{lipics}{\end{alphaenumerate}}{\end{itemize}}
\end{definition}

\iftoggle{lipics}{\paragraph*{Behavior Preservation}}{\subsubsection{Behavior Preservation}}

\begin{definition}[Refinement]
	A behavior $(\tau_1,b_1)$ \emph{refines} a behavior $(\tau_2,b_2)$ — denoted $(\tau_1, b_1) \refines (\tau_2, b_2)$ — if either $\tau_1=\tau_2 \wedge b_1=b_2$, or ${b_2 = \behErr} \wedge {\tau_2 \leq \tau_1}$ (with $\tau_2$ finite).
\end{definition}


\begin{definition}[Preservation]
	Let $\Beh_p$ and $\Beh_{p'}$ be two behavioral semantics.
	\begin{enumerate}
		\item $p \fwpreserv p'$ (forward preservation)
		$:\Leftrightarrow$ for every behavior $\Beh_p(\tau, b)$ there exists a behavior $\Beh_{p'}(\tau',b')$ such that $(\tau',b') \refines (\tau,b)$.
		\item $p \bwpreserv p'$ (backward preservation)
		$:\Leftrightarrow$ for every behavior $\Beh_{p'}(\tau', b')$ there exists a behavior $\Beh_{p}(\tau,b)$ such that $(\tau',b') \refines (\tau,b)$.
		\item $p \equiv p'$ (equivalence)
		$:\Leftrightarrow$ for all $(\tau,b)$, $\Beh_p(\tau,b) \Leftrightarrow \Beh_{p'}(\tau,b)$.
	\end{enumerate}
\end{definition}

If $p$ is compiled (or transformed) into $p'$, then $p \bwpreserv p'$ is the most desirable correctness property: every behavior of $p'$ is either also a behavior of $p$ or a refinement of a behavior of $p$ \cite{compcert-main}.
In particular, the compiler may replace undefined behavior in $p$ with some other behavior in $p'$, but it must not alter the observable trace up to the point where the undefined behavior occurs.
Refinement and the three forms of preservation are transitive.

Forward preservation of a transformation can be proven very naturally by showing implications between the three behavior predicates (see \cref{section-loop-trafos} for examples):

\begin{lemma}
\label[lemma]{lemma-make-forward-preservation}
	$p \fwpreserv p' \Leftrightarrow \left[ {(p \BSEval{t} \stateAbstr \Rightarrow p' \BSEval{t} \stateAbstr) \wedge (p \BSEvalinf{\tau} \Rightarrow p' \BSEvalinf{\tau}) \wedge (p \BSPartial{t} \Rightarrow p' \BSPartial{t})} \right]$. \coqed%
	\footnote{
		Detailed proofs can be found in our CompCert extension \cite{compcert-loop}.
	}
\end{lemma}

Luckily, the other forms of preservation can be derived from forward preservation alone:

\begin{theorem}
	\label{theorem-fw-to-bw}
	Let $p$ be receptive and $p'$ determinate. Then, ${(p \fwpreserv p')} \Rightarrow {(p \bwpreserv p')}$.
\end{theorem}
\begin{proof}[Proof Sketch]
	\localcoqed
	The goal is to show $\Beh_{p'}(\tau, b) \Rightarrow \Beh_p(\tau, b) \vee \exists t \prefix \tau, \Beh_p(t, \behErr)$.
	Instead we show a partial version of it: $p' \BSPartial t\ \Rightarrow p \BSPartial t \vee\ \exists t' \prefixStrict t, \Beh_p(t', \behErr)$.
	Using determinacy of $p'$ and compactness, this can be lifted to full behaviors, yielding the theorem.

	We prove the partial version via an induction on $|t|$. The base case $p \BSPartial \emptyTrace$ is given by property (1).
	So now let $p' \BSPartial {t \traceConcat e}$ for some event $|e|=1$.
	By property (2) and induction, either $p \BSPartial t$ or $\Beh_p(t',\behErr)$ for some $t' \prefixStrict t$.
	In the second case we are done.
	Else it is $p \BSPartial t$.
	Classically, either $\nexists \tilde e$ with $|\tilde e| = 1$ and $p \BSPartial {t \traceConcat {\tilde e}}$, in which case we get the desired $\Beh_p(t, \behErr)$, or otherwise we have $p \BSPartial {t \traceConcat {\tilde e}}$ and, using \cref{lemma-make-forward-preservation}, $p' \BSPartial {t \traceConcat {\tilde e}}$.
	Determinacy of $p'$ then yields $e \asymp \tilde e$ and by receptiveness of $p$ we conclude $p \BSPartial {t \traceConcat e}$, as required.
\end{proof}

\begin{theorem}
	\label{theorem-fw-to-equiv}
	Let $p$ be determinate. Then, $({p \fwpreserv p'}) \wedge {(p' \fwpreserv p)} \Rightarrow {p \equiv p'}$. \coqed
\end{theorem}
\begin{proof}[Proof Sketch]
	\localcoqed
	Using determinacy of $p$, we first show the ``half equivalence'' $\Beh_p(\tau,b) \Rightarrow \Beh_{p'}(\tau,b)$.
	$p' \fwpreserv p$ is enough to carry determinacy over from $p$ to $p'$, and therefore we can symmetrically show the other half equivalence and are done.
	
	Assume $\Beh_p(\tau,b)$.
	By $p \fwpreserv p'$ it is $\Beh_{p'}(\tau',b')$ for some $(\tau',b') \refines (\tau,b)$,
	and by $p' \fwpreserv p$ it is $\Beh_p(\tau'',b'')$ for some $(\tau'',b'') \refines (\tau',b') \refines (\tau,b)$.
	All these three tuples must be equal however due to determinacy of $p$, as $\tau$ is a prefix of $\tau''$ (because $\refines$ is transitive).
\end{proof}

\cref{theorem-fw-to-bw} is the behavioral analogue of Theorem 4.5 in \cite{compcert-tso}, which CompCert uses to transform a small-step forward simulation into a backward simulation.
Notably, it has the same premises: receptiveness of the source language and determinacy of the target language.
	\subsection{Small-Step and Big-Step Semantics}
\label{section-small-bigstep}

Our goal is it to bring small-step and big-step semantics into above behavioral framework.

For a small-step semantics $\SSStep{}$ we can use its small-step behavior to define the behavior predicates:
given a program $p$ with initial state $\stateAbstr_p$, define
$p \BSPartial t\ :\Leftrightarrow \exists \stateAbstr, \stateAbstr_p \SSStar t \stateAbstr$
and
${p \BSEval t \stateAbstr} :\Leftrightarrow \stateAbstr_p \SSStar t \stateAbstr_{fin}$ (with $\stateAbstr_{fin}$ final state)
and
$p \BSEvalinf \tau\ :\Leftrightarrow \stateAbstr_p \SSForever \tau$,
where $\SSForever{\tau}$ is a coinductive predicate of divergence as defined in \cref{section-coinductive}.

\begin{lemma}
When $\SSStep{}$ is determinate, above definition yields a determinate behavioral semantics whose behavior coincides with the small-step behavior of $\SSStep{}$.
When $\SSStep{}$ is additionally receptive, the behavioral semantics is also receptive. \coqed
\end{lemma}

Determinacy is required because otherwise, partial execution alone is not strong enough to detect going wrong: there could be two executions 
	$\stateAbstr_p \SSStar t \stateAbstr_1 \not \SSStep{}$ and 
	$\stateAbstr_p \SSStar {t \traceConcat e} \stateAbstr_2$,
	in which case $\Beh_p(t, \behErr)$ would incorrectly fail to hold.
That the behavior coincides with the small-step behavior as defined by CompCert in case of divergence is non-trivial -- see \cref{thm-foreverx-equiv}. ~\\

A big-step semantics as described by Leroy \cite{coind-bs-leroy} is defined by an inductive terminaton judgment and a coinductive divergence judgment.
Two adaptations are required to fit it into the behavioral framework:
introducing a partial judgment and changing the divergence judgment to allow a general trace.
Both of these modifications are presented in \cref{section-coinductive}.
A main benefit of behavioral big-step semantics is that we do not need an additional, complicated going wrong judgment but can derive it from the others which are generally simpler.
However, showing the required properties of the partial predicate is nontrivial for a big-step semantics.
We can solve this be resorting to an equivalent small-step semantics: if we can show soundness and completeness w.~r.~t.~a behavioral small-step semantics, the respective partial properties of the big-step semantics follow directly.

If $p$ is a small-step semantics and $p^{BS}$ a big-step semantics, both interpreted as behavioral semantics, soundness of $p^{BS}$ w.~r.~t.~$p$ is exactly the conjunction of the three statements
$(p^{BS} \BSEval t \stateAbstr \Rightarrow p \BSEval t \stateAbstr)$ and $(p^{BS} \BSEvalinf \tau \Rightarrow p \BSEvalinf \tau)$ and $(p^{BS} \BSPartial t \Rightarrow p \BSPartial t)$.
By \cref{lemma-make-forward-preservation}, this yields $p^{BS} \fwpreserv p$.
Similarly, completeness gives $p \fwpreserv p^{BS}$.
If $p$ is determinate, then \cref{theorem-fw-to-equiv} concludes the equivalence $p \equiv p^{BS}$.

\begin{lemma}
	\label{lemma-det-rec-preserved}
	Determinacy and receptiveness are preserved under equivalence. \coqed
\end{lemma}
	\subsection{Pipeline Integration}
\label{section-pipeline-integration}

We describe the modifications we made to CompCert in order to integrate behavioral semantics, noting that the approach also applies to other small-step transformation pipelines beyond CompCert.
\cref{fig-compcert-pipeline} shows a typical compiler pipeline at the top, starting from the C source program $p_c$ and ending in the assembly program $p_{asm}$.
Each transformation establishes a small-step simulation \cite{compcert-3}, which is then lifted to the entire pipeline.
From the small-step simulation between $p_c$ and $p_{asm}$ one obtains forward preservation $p_c \fwpreserv p_{asm}$ and, using additional properties of the semantics, ultimately backward preservation $p_c \bwpreserv p_{asm}$, the main theorem of a verified compiler.
Note that $p_c$ and $p_{asm}$ are classical small-step semantics and not behavioral semantics — the notations $\fwpreserv$ and $\bwpreserv$ are therefore used here in the classical sense to refer to small-step behaviors.

\newcommand{\pcmin}{p_{cmin}}
\newcommand{\pcminl}{{\tilde p}_{cmin}}
\newcommand{\pcminBS}{p^{BS}_{cmin}}
\newcommand{\pcminlBS}{{\tilde p}^{BS}_{cmin}}

\begin{figure}
	$$
		p_c \rightarrow p_{clight} \rightarrow \cdots \rightarrow p_{cmin} \rightarrow \cdots \rightarrow p_{asm}
	$$
	\vspace{-4mm}

	\begin{centeredmath}
		$$
		\begin{array}{ccccc}
			&\multicolumn{3}{c}{
				\overbrace{\phantom{\quad \pcminBS \longrightarrow \dots \longrightarrow \pcminlBS}}
				^{{\pcminBS \bwpreserv \pcminlBS}}
			} & \\[-1.3em]
			
			\phantom{heydutest}&\coord{start}\pcminBS & \underset{
				\mathclap{\smash[b]{\substack{\text{bigstep loop}\\\text{transformations}}}}
			}{\longrightarrow \dots \longrightarrow} & \pcminlBS &\\
			
			& {\Big\uparrow \mathrlap{\equiv}} & & {\mathllap{\equiv} \Big\downarrow} &\\
			
			\multicolumn{2}{r}{
				\underbrace{
					p_c \overset{
						\smash[t]{\substack{\text{smallstep}\\\text{simulation}\vspace{1mm}}}
					}{\rightarrow \dots \rightarrow} \pcmin
				}_{p_c \bwpreserv \pcmin}
			} & & \multicolumn{2}{l}{
				\underbrace{
					\pcminl \coord{end} \overset{
						\smash[t]{\substack{\text{smallstep}\\\text{simulation}\vspace{1mm}}}
					}{\rightarrow \dots \rightarrow} \tilde p_{asm}
				}_{\pcminl \bwpreserv \tilde p_{asm}}
			}
		\end{array}
		$$
	\end{centeredmath}
	\begin{tikzpicture}[remember picture, overlay]
		\draw[red, dashed, rounded corners]
		([yshift=6.4ex, xshift=-0.5ex]start)
		rectangle
		([yshift=-2ex, xshift=0.5ex]end);
	\end{tikzpicture}
	
	\captionsetup{width=.85\linewidth}
	\caption{
		At the top, the classical CompCert pipeline: $p_c$ is compiled into $p_{asm}$.
		At the bottom, the new pipeline that performs loop transformations on \Cminor.
		The part inside the dashed rectangle uses the behavioral framework.
	}
	\label{fig-compcert-pipeline}
\end{figure}
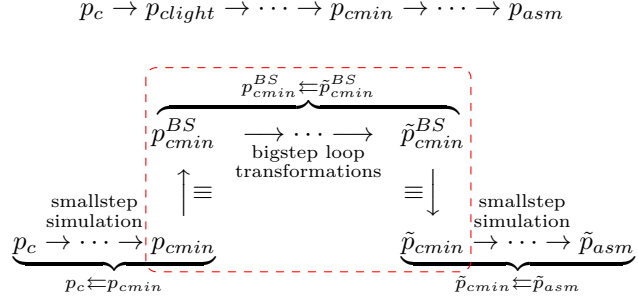

\cref{fig-compcert-pipeline} shows our pipeline at the bottom.
We use the (determinate and receptive) intermediate language \Cminor, whose small-step and big-step semantics we interpret as behavioral semantics.
At \Cminor, we split CompCert's pipeline into two parts and move to the big-step semantics in order to carry out various loop transformations ${\pcminBS \rightarrow \dots \rightarrow \pcminlBS}$.
Afterwards, we switch back to the small-step semantics.
The two parts of the small-step pipeline each yield a small-step simulation, from which a backward preservation follows.
By \cref{theorem-fw-to-bw}, a backward preservation follows for the loop transformations between $\pcminBS$ and $\pcminlBS$.
The equivalence between small-step and big-step \Cminor follows from soundness and completeness together with \cref{theorem-fw-to-equiv}.
This yields the chain
\begin{centeredmath}
	\begin{gather*}
		p_c \bwpreserv \pcmin \equiv \pcminBS \bwpreserv \pcminlBS \equiv \pcminl \bwpreserv \tilde p_{asm}, 
		\text{ and therefore: } p_c \bwpreserv \tilde p_{asm}.
	\end{gather*}
\end{centeredmath}

This restores CompCert's main theorem.
Here again, for small-step semantics, $\bwpreserv$ refers to small-step behaviors.
We were also able to establish all other main theorems of CompCert for our extended compiler chain. These include results on separate compilation via linking \cite{cc-linking-1, cc-linking-2} as well as statements about the behavior with respect to the alternative \texttt{Cstrategy} semantics that is implemented in CompCert. Further details can be found in our development.
	\section{Coinductive Divergence with General Trace}
\label{section-coinductive}

\begin{figure}
	\centering
	
	\rulefigdouble{0.17}
	{
		s \SSStep t s' \hspace{3mm} \forever\omega s'
	}{
		\forever{t \cdot \omega} s
	}
	\rulefigdouble{0.24}
	{
		s \SSStep \emptyTrace s' \hspace{3mm} \foreversilent s'
	}{
		\foreversilent s
	}
	\rulefig{0.25}
	{
		s \SSStar t s' \hspace{3mm} \foreversilent s'
	}{
		\foreverfinite t s
	}
	\rulefigdouble{0.28}
	{
		s \SSStar t s' \hspace{3mm} t \neq \emptyTrace \hspace{3mm} \reacts\omega s'
	}{
		\reacts{t \cdot \omega} s
	}
	
	\captionsetup{width=.99\linewidth}
	\caption{Various definitions of small-step divergence as found in CompCert.
	Here and in the following, we use single lines for inductive judgments and double lines for coinductive judgments.}
	\label{fig-forever}
\end{figure}

Given a small-step semantics $\SSStep{}$, \cref{fig-forever} presents on the left the classical coinductive definition of divergence ($\forever{\omega}$) producing an infinite trace $\omega$.
The problem of this definition is that it cannot determine whether $\omega$ is generated coinductively in its entirety or whether only a finite prefix $t \prefixStrict \omega$ is ever produced.
For instance, an empty \texttt{while true}-loop would diverge with any infinite $\omega$.
This issue can be addressed by resorting to multiple different coinductive definitions for silent ($\foreversilent$), finite ($\foreverfinite t$) and infinite ($\reacts \omega$) divergence which are also shown in \cref{fig-forever}.
$\forever \omega$ cannot decide between finite and infinite divergence; however, it holds that:
$$
\forever \omega s \Rightarrow \reacts \omega s \vee \exists t \prefixStrict \omega, \foreverfinite t s. \eqno\hbox{$(*)$} 
$$

CompCert follows this approach and uses both $\foreverfinite{t}$ and $\reacts{\omega}$ to define the divergence behavior of a small-step semantics.
However, this method cannot be transferred to a coinductive big-step semantics: first, all its divergence rules would have to be duplicated, and second, it is unclear how to transfer the $\SSStar t$ from $\reacts{t \traceConcat \omega}$ with its $t \neq \emptyTrace$ condition into big-step rules.

We therefore combine these definitions into \emph{a single} one that produces a general%
\footnote{
Our general trace is simply the sum type of finite and infinite traces. 
This works because we can distinguish classically between whether the trace is finite or infinite.}
(i.~e., finite or infinite) trace $\tau$ and that ensures that this trace is indeed produced completely.

Our judgment $\SSForeverx \tau n$, whose definition is shown in \cref{fig-foreverx}, additionally carries a natural number $n \in \mathbb{N}_0$ to ensure that after at most $n$ steps a non-empty prefix of $\tau$ has been produced.
This continues until $\tau$ is exhausted in the case where $\tau$ is finite, or indefinitely otherwise.
For this purpose, we define $guard\ \tau\ n\ t\ m := {(\tau = \emptyTrace) \vee (t = \emptyTrace \Rightarrow m < n)}$.
The concrete value of $n$ has no semantic relevance; we therefore occasionally write only $\SSForever \tau$.~\\

\begin{figure}
	\centering
	
	\rulefigdouble{0.6}
	{
		s \SSStep t s' \hspace{4mm} s' \SSForeverx \tau m \hspace{4mm} guard\ \tau\ n\ t\ m
	}{
		s \SSForeverx{t \cdot \tau} n
	}
	\captionsetup{width=.8\linewidth}
	\caption{Our definition of small-step divergence with a general trace.
		Here, $m, n \in \mathbb{N}_0$ and
		$guard\ \tau\ n\ t\ m := {(\tau = \emptyTrace) \vee (t = \emptyTrace \Rightarrow m < n)}$.
	}
	\label{fig-foreverx}
\end{figure}
\begin{theorem}
	\label[theorem]{thm-foreverx-equiv}
	\vbox{ 
	$
	s \SSForever{\tau} \Rightarrow
	\begin{cases}
		\foreverfinite \tau s & \text{for $\tau$ finite,} \\
		\reacts{\tau} s & \text{for $\tau$ infinite.}
	\end{cases} \\
	$

	If $\SSStep{}$ is determinate\footnotemark{}, equivalence holds. \coqed
	}
	\footnotetext{
		Given $\SSStar{t}$ with $t \neq \emptyTrace$, determinacy is needed to \emph{uniquely} identify the step count at which the first event occurs, which is required for a coinductive definition of $\SSForever{\tau}$.
	}
\end{theorem}

A coinductive big-step semantics, as originally defined by Cousot and Cousot \cite{coind-bs-cousot} and extended by Leroy \cite{coind-bs-leroy} has a coinductive judgment for divergence with an infinite trace.
This gives rise to the same problem as in small-step semantics:
the empty loop can diverge with any trace, $(\LOOP \SKIP) \BSEvalinf \omega$.
We can however employ the same trick as above.
The following figure shows, on the left, two coinductive definitions in a big-step semantics that characterize the divergence of a loop (the full semantics is given in \cref{section-cminor}):

{
\setlength{\textfloatsep}{0pt} 
\setlength{\floatsep}{0pt}     
\setlength{\intextsep}{0pt}

\begin{figure}[H]
\coord{fig}
\centering
\rulefigdouble{.45} {
		\BSState {s} {\stateAbstr} \BSEvalinf {\omega}
	}{
		\BSState {\LOOP {s}} {\stateAbstr} \BSEvalinf {\omega}
	}
\rulefigdouble{.5} {
		\BSState {s} {\stateAbstr} \BSEvalinfx {\tau} m
		\hspace{3mm}
		guard\ \tau\ n\ \emptyTrace\ m
	}{
		\BSState {\LOOP {s}} {\stateAbstr} \BSEvalinfx {\tau} n
	}
\rulefigdouble{.45} {
		\BSState {s} {\stateAbstr} \BSEval t {\stateAbstr'}
		\hspace{3mm}
		\BSState {\LOOP {s}} {\stateAbstr'} \BSEvalinf {\omega}
	}{
		\BSState {\LOOP {s}} {\stateAbstr} \BSEvalinf {t \traceConcat \omega}
	}
	\rulefigdouble{.5} {
		\BSState {s} {\stateAbstr} \BSEval t {\stateAbstr'}
		\hspace{3mm}
		\BSState {\LOOP {s}} {\stateAbstr'} \BSEvalinfx {\tau} m
		\hspace{3mm}
		guard\ \tau\ n\ t\ m
	}{
		\BSState {\LOOP {s}} {\stateAbstr} \BSEvalinfx {t \traceConcat \tau} n
	}
\end{figure}
\begin{tikzpicture}[remember picture, overlay]
    \node at ([yshift=-.5cm,xshift=.43\linewidth] fig) {$\xrightarrow{\phantom{abcde}}$};
\end{tikzpicture}
}
A loop can diverge either when its body diverges ($\BSEvalinf \omega$) or when its body terminates ($\BSEval t$) and the loop diverges afterwards.
Intepreted coinductively, this also includes indefinite execution of a terminating loop body.
We update each of these judgments, seen on the right, by adding the just introduced $guard$ condition.
If the rule application cannot itself produce a trace, we use $\emptyTrace$ inside the guard, otherwise we use the produced trace $t$.
This guarantees that after at most $n$ applications of the divergence rules, a nonempty prefix of $\tau$ has been produced, except when $\tau$ is empty, and therefore enforces the full production of the now general trace $\tau$.
As in the small-step case, the actual value of $n \in \mathbb{N}_0$ is irrelevant and therefore we will interchangeably use $\BSEvalinfx \tau n$ and $\BSEvalinf \tau$.

While we do not need to establish any relation between the new and the old big-step judgments, it is an easy coinduction to show the following:
$$
\forall \tau \prefix \omega,\ \BSState s \stateAbstr \BSEvalinfx{\tau} n \Rightarrow \BSState s \stateAbstr \BSEvalinf{\omega}.
$$
Proving the other direction (given $\omega, \exists \tau \prefix \omega$) is also possible for determinate semantics such as $\Cminor$.
The proof is not direct however but is done by going through new and old small-step semantics that are equivalent to the respective big-step semantics and applying above equation $(*)$ and \cref{thm-foreverx-equiv} to convert between them.
	\section{Cminor}
\label{section-cminor}

\Cminor is an early intermediate representation of CompCert which consists of expressions, statements and functions \cite{cminor}.
A function is not defined by a control flow graph, but instead by a single, inductively defined statement.
We present an excerpt of the modified \Cminor (omitting internal function calls, memory and other things for brevity) which we use in the remainder of this paper.
The following statements exist:
\begin{centeredmath}
	\begin{alignat*}{2}
		stmt :=\ &\SKIP &&\\
		|\ &\STORE {reg} {expr} &&\text{register operation} \\
		|\ &\IFTHENELSE {expr} {stmt} {stmt} \hspace{1cm} &&\text{if-then-else} \\
		|\ &\SEQ {stmt} {stmt} &&\text{sequence} \\
		|\ &\LOOP {stmt} &&\text{loop} \\
		|\ &\BLOCK {stmt} &&\text{wrap in block} \\
		|\ &\EXIT {nat} &&\text{exit $n+1$ enclosing blocks} \\
		|\ &\EXTCALL {function} && \text{call external function}
	\end{alignat*}
\end{centeredmath}

The $\EXTCALL$ statement is the only one that produces an event.
Blocks are required around loops to allow exiting out of them via the $\EXIT {\!}$ statement.
For example, the following statement describes a loop that executes $stmt$ 5 times:
\begin{centeredmath}
	\begin{alignat*}{1}
		\SEQ {\STORE i 0}
		\llbracket \LOOP \big(\
			& (i <^? 5) \\
			?\ &\SEQ {stmt} {(\STORE {i} {i + 1})} \\
			:\ &\EXIT 0\ \big) \rrbracket
	\end{alignat*}
\end{centeredmath}


Both small- and big-step semantics are presented in \cref{fig-cminor-sem}.
A small-step state looks like $\SSState {stmt} {cont} {\stateAbstr}$, where $cont$ is the continuation that happens after $stmt$ has been executed and $\stateAbstr$ contains the global environment and the register state.
A continuation is either $\Kstop$ (execution ends), $\Kseq {stmt} {cont}$ (first $stmt$ is executed and then $cont$) or $\Kblock {cont}$ ($cont$ is wrapped inside a block).%
\footnote{
	In CompCert \Cminor, the stack trace is also fully encoded in the continuation.
}


The big-step semantics consists of two judgments: an inductive judgment $\BSState {stmt} \stateAbstr \BSEval t out$ and a coinductive judgment $\BSState {stmt} \stateAbstr \BSEvalinfx \tau n$.
A statement executes with an \emph{outcome} that is either $\outOk {\stateAbstr'}$ (normal execution with final state $\stateAbstr'$), $\outExit n {\stateAbstr'}$ (exit $n+1$ enclosing blocks where $n \in \mathbb{N}_0$) or $\outPartial$.

We have modified the original big-step semantics in two ways: first, we added $guard$ statements to each divergence rule to allow divergence with a general trace, as described in \cref{section-coinductive}.
Second, we added two rules with a $\outPartial$ outcome, thereby combining both partial and full evaluation into a single judgment.

Notice that, to achieve a $\outPartial$ outcome, the full trace must have been produced before an evaluation can be ``truncated'' using the rule $s \BSEval \emptyTrace \outPartial$.
From the definition it follows that any full evaluation induces a partial evaluation with the same trace.
All the other required propreties for a behavioral semantics follow after soundness and completeness w.~r.~t.~the small-step semantics have been established, as described in \cref{section-small-bigstep}.

\begin{figure}
	\small

	\centering
  \begin{framed}
		\begin{subfigure}{\linewidth}
		  \centering
			\rulefig{.48} {
				\istrue (\eval \stateAbstr e)
			}{
				\SSState {\IFTHENELSE e {s_1} {s_2}} k \stateAbstr \SSStep \emptyTrace \SSState {s_1} k \stateAbstr
			}
			\rulefig{.48} {
				\isfalse (\eval \stateAbstr e)
			}{
				\SSState {\IFTHENELSE e {s_1} {s_2}} k \stateAbstr \SSStep \emptyTrace \SSState {s_2} k \stateAbstr
			}

			\rulefig{.48} {
				\eval \stateAbstr e = v
			}{
				\SSState {\STORE r e} k \stateAbstr \SSStep {\emptyTrace} \SSState \SKIP k {\stateAbstr[r := v]}
			}
			\rulefig{.48} {
				\extcall f \stateAbstr = \stateAbstr'
			}{
				\SSState {\EXTCALL f} k \stateAbstr \SSStep{\traceCALL f {$\stateAbstr$}} \SSState \SKIP k {\stateAbstr'}
			}

			\rulefignoline{.4} {
				\SSState {\LOOP s} k \stateAbstr \SSStep \emptyTrace \SSState s {\Kseq {(\LOOP s)} k} \stateAbstr
			}
			\rulefignoline{.4} {
				\SSState {\SEQ {s_1} {s_2}} k \stateAbstr \SSStep \emptyTrace \SSState {s_1} {\Kseq {s_2} k} \stateAbstr
			}

			\vspace{2mm}

			\rulefignoline{.4} {
				\SSState \SKIP {\Kseq s k} \stateAbstr \SSStep \emptyTrace \SSState s k \stateAbstr
			}
			\rulefignoline{.44} {
				\SSState \SKIP {\Kblock k} \stateAbstr \SSStep \emptyTrace \SSState \SKIP k \stateAbstr
			}

			\rulefignoline{.4} {
				\SSState {\BLOCK s} k \stateAbstr \SSStep \emptyTrace \SSState s {\Kblock k} \stateAbstr
			}
			\rulefignoline{.45} {
				\SSState {\EXIT n} {\Kseq s k} \stateAbstr \SSStep \emptyTrace \SSState {\EXIT n} k \stateAbstr
			}

			\rulefignoline{.43} {
				\SSState {\EXIT 0} {\Kblock k} \stateAbstr \SSStep \emptyTrace \SSState \SKIP k \stateAbstr
			}
			\rulefignoline{.52} {
				\SSState {\EXIT\ (n+1)} {\Kblock k} \stateAbstr \SSStep \emptyTrace \SSState {\EXIT n} k \stateAbstr
			}
		\caption*{Small-step judgment $\SSStep{t}$}
		\vspace{-2mm}
		\end{subfigure}
	\end{framed}

	\begin{framed}
		\vspace{-4mm}
		\centering
		\begin{subfigure}{\textwidth}
			\centering
			\rulefignoline{.25} {
				\forall s: s \BSEval \emptyTrace \outPartial
			}
			\rulefignoline{.2} {
				\BSState \SKIP \stateAbstr \BSEval\emptyTrace \outOk \stateAbstr
			}
			\rulefignoline{.3} {
				\BSState {\EXIT n} \stateAbstr \BSEval\emptyTrace \outExit n \stateAbstr
			}

			\rulefig{.4} {
				\istrue (\eval \stateAbstr e)
				\hspace{3mm}
				\BSState {s_1} \stateAbstr \BSEval t o
			}{
				\BSState {\IFTHENELSE e {s_1} {s_2}} \stateAbstr \BSEval t o
			}
			\rulefig{.4} {
				\isfalse (\eval \stateAbstr e)
				\hspace{3mm}
				\BSState {s_2} \stateAbstr \BSEval t o
			}{
				\BSState {\IFTHENELSE e {s_1} {s_2}} \stateAbstr \BSEval t o
			}

			\rulefig{.4} {
				\eval \stateAbstr e = v
			}{
				\BSState {\STORE r e} \stateAbstr \BSEval{\emptyTrace} \outOk \stateAbstr[r := v]
			}
			\rulefig{.5} {
				\extcall f \stateAbstr = \stateAbstr'
			}{
				\BSState {\EXTCALL f} \stateAbstr \BSEval{\traceCALL f {$\stateAbstr$}} \outOk \stateAbstr' \vee \outPartial^{(*)}
			}

			\rulefigannot{.4}{\LOOP_{exit}} {
				\BSState {s} {\stateAbstr} \BSEval {t} o
				\hspace{5mm}
				o \neq \outOk \_
			}{
				\BSState {\LOOP {s}} {\stateAbstr} \BSEval {t} o
			}
			\rulefigannot{.45}{\LOOP_{cont}} {
				\BSState {s} {\stateAbstr_1} \BSEval {t_1} \outOk {\stateAbstr_2}
				\hspace{5mm}
				\BSState {\LOOP s} {\stateAbstr_2} \BSEval {t_2} o
			}{
				\BSState {\LOOP {s}} {\stateAbstr_1} \BSEval {t_1 \traceConcat t_2} o
			}

			\rulefigannot{.4}{\kseq_{exit}} {
				\BSState {s_1} {\stateAbstr} \BSEval {t} o
				\hspace{5mm}
				o \neq \outOk \_
			}{
				\BSState {\SEQ {s_1} {s_2}} {\stateAbstr} \BSEval {t} o
			}
			\rulefigannot{.45}{\kseq_{cont}} {
				\BSState {s_1} {\stateAbstr_1} \BSEval {t_1} \outOk {\stateAbstr_2}
				\hspace{5mm}
				\BSState {s_2} {\stateAbstr_2} \BSEval {t_2} o
			}{
				\BSState {\SEQ {s_1} {s_2}} {\stateAbstr_1} \BSEval {t_1 \traceConcat t_2} o
			}

			\rulefigannot{.4}{\BLOCK{}_{exit}^1} {
				\BSState s \stateAbstr \BSEval t \outOk \stateAbstr' \vee \outExit 0 \stateAbstr'
			}{
				\BSState {\BLOCK s} \stateAbstr \BSEval t \outOk \stateAbstr'
			}
			\rulefigannot{.4}{\BLOCK{}_{exit}^2} {
				\BSState s \stateAbstr \BSEval t \outExit {(n+1)} \stateAbstr'
			}{
				\BSState {\BLOCK s} \stateAbstr \BSEval t \outExit n \stateAbstr'
			}

		\caption*{Inductive partial and full execution judgment $\BSEval{t}$ \hfill \scriptsize $(*)$ These are two separate rules}
		\vspace{-2mm} 
	\end{subfigure}
	\hbox to \linewidth {\dotfill}
	\vspace{1mm}
	\begin{subfigure}{\linewidth}
			
			\rulefigdouble{.48} {
				\istrue (\eval \stateAbstr e)
				\hspace{3mm}
				\BSState {s_1} \stateAbstr \BSEvalinfx \tau m
			}{
				\BSState {\IFTHENELSE e {s_1} {s_2}} \stateAbstr \BSEvalinfx \tau n
			}
			\rulefigdouble{.48} {
				\isfalse (\eval \stateAbstr e)
				\hspace{3mm}
				\BSState {s_2} \stateAbstr \BSEvalinfx \tau m
			}{
				\BSState {\IFTHENELSE e {s_1} {s_2}} \stateAbstr \BSEvalinfx \tau n
			}
			
			\rulefigdoubleannot{.3}{\BLOCK{}_{div}} {
				\BSState s \stateAbstr \BSEvalinfx \tau m
			}{
				\BSState {\BLOCK s} \stateAbstr \BSEvalinfx \tau n
			}
			\rulefigdoubleannot{.3}{\kseq_{div}} {
				\BSState {s_1} {\stateAbstr} \BSEvalinfx {\tau} m
			}{
				\BSState {\SEQ {s_1} {s_2}} {\stateAbstr} \BSEvalinfx {\tau} n
			}
			\rulefigdoubleannot{.3}{\LOOP_{div}} {
				\BSState {s} {\stateAbstr} \BSEvalinf {\tau} m
			}{
				\BSState {\LOOP {s}} {\stateAbstr} \BSEvalinf {\tau} n
			}
			\rulefigdoubleannot{.48}{\kseq'_{cont}} {
				\BSState {s_1} {\stateAbstr_1} \BSEval {t} \outOk {\stateAbstr_2}
				\hspace{5mm}
				\BSState {s_2} {\stateAbstr_2} \BSEvalinfx {\tau} m
			}{
				\BSState {\SEQ {s_1} {s_2}} {\stateAbstr_1} \BSEvalinfx {t \traceConcat \tau} n
			}
			\rulefigdoubleannot{.48}{\LOOP'_{cont}} {
				\BSState {s} {\stateAbstr_1} \BSEval {t} \outOk {\stateAbstr_2}
				\hspace{5mm}
				\BSState {\LOOP s} {\stateAbstr_2} \BSEvalinfx {\tau} m
			}{
				\BSState {\LOOP {s}} {\stateAbstr_1} \BSEvalinfx {t \traceConcat \tau} n
			}
		\caption*{Coinductive divergence judgment $\BSEvalinfx{\tau}{n}$. The $guard$ premise is omitted in each rule.}
		\vspace{-2mm}
	\end{subfigure}
	\end{framed}
	\caption{Small-step and big-step semantics of \Cminor (excerpt)}
	\label{fig-cminor-sem}
\end{figure}

	\subsection{Soundness and Completeness}

\icfponly{\paragraph*{Soundness}}

\iftoggle{lipics}{\textbf{Soundness}}{Soundness} of big-step \Cminor means that every terminating, partial or diverging big-step judgment yields a respective small-step execution.
This can be proven inductively resp.~coinductively and has already been done in CompCert for the classical \Cminor.
As it is not strong enough to distinguish between finite-trace and an infinite-trace divergence, the soundness statement was accordingly weak, only giving a $\forever \omega$ result.
Using our general big-step divergence, we were able to readily strengthen it into $\SSForever \tau$, which yields the correct small-step behavior by \cref{thm-foreverx-equiv}.
Also, we could almost trivially prove the partial case by integrating it into the terminating case.


\icfponly{\paragraph*{Completeness of Termination}}

\iftoggle{lipics}{\textbf{Completeness}}{Completeness} for termination means that, given any terminating small-step execution
$\SSState s k \stateAbstr \SSStar{t} \SSState \SKIP \Kstop {\stateAbstr_{fin}}$,
the statement $s$ must evaluate successfully, so there must be a ``terminal'' state of $s$ in between the original and the final state that corresponds to a big-step evaluation of $s$:
$$
\SSState s k \stateAbstr \SSStar{t_1} \SSState \SKIP k {\stateAbstr'} \SSStar{t_2} \SSState \SKIP \Kstop {\stateAbstr_{fin}}\ \text{ with } t = t_1 \traceConcat t_2 \text{ and } \BSState s \stateAbstr \BSEval{t_1} \outOk {\stateAbstr'}.%
$$

This can be proven by doing a reverse induction on the (finite) execution.
In a very similar fashion we can show that \emph{any} small-step execution $\SSState s k \stateAbstr \SSStar{t} \SSState {s'} {k'} {\stateAbstr'}$ either corresponds to an ongoing partial evaluation of $s$ ($s \BSEval t \outPartial$) or that a terminating state of $s$ exists within the execution together with a full big-step evaluation of $s$, just as above.
Applying both statements to the initial state gives us the desired completeness of $\BSEval{}$ and $\BSPartial{}$.

\icfponly{\paragraph*{Completeness of Divergence}}

If a \iftoggle{lipics}{\textbf{diverging}}{diverging} small-step execution visits the state $\SSState s k \stateAbstr$, then either $s$ itself causes the divergence or $s$ terminates and something inside $k$ diverges.
Therefore, we want to show:
\begin{centeredmath}
	\begin{gather*}
		\text{Given } \SSState s k \stateAbstr \SSForeverx \tau n,
		\text{ either }
		\left( \BSState s \stateAbstr \BSEvalinfx \tau n \right)
		\text{ or }
		\left( \exists (t\leq\tau), (out\neq\outPartial): \BSState s \stateAbstr \BSEval t out \right)
	\end{gather*}
\end{centeredmath}
This statement cannot yet be proven coinductively because of the $\vee$ and $\exists$ in its conclusion.
Negating the second option of the conclusion and bringing it to the left side as a premise gives the required form however.
Applying it to the initial state, together with soundness and determinacy, we were able to conclude completeness of $\BSEvalinf{}$.
As \Cminor is determinate and receptive, \cref{theorem-fw-to-equiv} gives behavioral equivalence and \cref{lemma-det-rec-preserved} yields determinacy and receptiveness for big-step $\Cminor$.%

	\section{Loop Transformations}
\label{section-loop-trafos}

\Cminor is a perfect%
\footnote{
	The big-step semantics naturally does not support goto statements, however.
	Our development therefore gives the user the option to skip the big-step stage if goto statements are present.
}
candidate for reasoning in big-step semantics due to its inductive nature and especially its inductive loop statement: transformations are straightforward to define and their big-step proofs are often quite natural.
We show and discuss here some of the loop tranformations that we have implemented.

\subsection{Loop Unswitching}

Loop unswitching is a transformation that hoists branches out of a loop when they are independent of the loop body.
We restrict our attention to the simple case in which the loop body immediately starts with the independent branch and is split into two parts by it.
We call such loops ``relevant''.

The transformation $f: stmt \to stmt$ is defined inductively over statement structure:
\begin{alignat*}{2}
	f (\LOOP (\IFTHENELSE{c}{s_1}{s_2})) &= \IFTHENELSE{c}{\LOOP (f s_1)}{\LOOP (f s_2)} \hspace{1em} &&\text{if } \indep c {s_1} \wedge \indep c {s_2} \\
	f (\LOOP s) &= \LOOP (f s) &&\text{for other loops} \\
	f (\IFTHENELSE{c}{s_1}{s_2}) &= \IFTHENELSE{c}{f s_1}{f s_2} \\
	f (\SEQ {s_1} {s_2}) &= \SEQ {f s_1} {f s_2} \\
	f (\BLOCK s) &= \BLOCK {f s} \\
	f s &= s &&\text{for all other $s$}.
\end{alignat*}
Thus, $f$ transforms all relevant loops, including nested ones, regardless of how deep they are inside the structure of the statement.
The predicate $\indep e s$ expresses that the statement $s$ does not write to any registers that are used by the expression $e$.

\begin{lemma}[indep\_spec]
	\label[lemma]{lemma-indep-spec}
	$\indep c s \wedge \BSState s \stateAbstr \BSEval t \outOk {\stateAbstr'} \Rightarrow \eval{c}{\stateAbstr} = \eval{c}{\stateAbstr'}$.
\end{lemma}
This is the main lemma relating $indep$ and $eval$ and is proven via induction on $\BSEval t$.

To show the correctness of the transformation -- meaning, $f$ induces a forward preservation $p \fwpreserv f p$ between a program and its transformed version -- using big-step semantics, the two statements $\BSState{s}{\stateAbstr} \BSEval{t} out \Rightarrow \BSState{f s}{\stateAbstr} \BSEval{t} out$ and $\BSState{s}{\stateAbstr} \BSEvalinf{\tau} \Rightarrow \BSState{f s}{\stateAbstr} \BSEvalinf{\tau}$ suffice.
To show $\BSState{s}{\stateAbstr} \BSEval{t} out \Rightarrow \BSState{f s}{\stateAbstr} \BSEval{t} out$ we use structural induction on $\BSEval t$, which directly yields the result for every $s$ that is \emph{not} a special loop.

So now consider a special loop $\BSState {\LOOP(\IFTHENELSE{c}{s_1}{s_2})} \stateAbstr \BSEval t o$.
We need to show $\BSState {\IFTHENELSE{c}{\LOOP (f s_1)}{\LOOP (f s_2)}} \stateAbstr \BSEval t o$.
The assumption can arise in two cases: the loop is exited during the first iteration via the rule $[\LOOP_{exit}]$, or it is exited in a subsequent iteration via the rule $[\LOOP_{cont}]$.
The first case is treated in the following diagram:

\def\defaultHypSeparation{\hskip .1in}

\begin{table}[H]
	\centering
	\begin{tabular}{cc}
		\makecell[c]{
			$\left\uparrow \vbox to 13mm{} \right.$
			\hspace{-1em}
			\AxiomC{$\istrue (\eval c \stateAbstr)$}
			\AxiomC{$\BSState{f s_1}{\stateAbstr} \BSEval t o$}
			\RightLabel{\scriptsize wlog}
			\BinaryInfC{$\BSState {\overbrace{f (\IFTHENELSE{c}{s_1}{s_2})}^{\IFTHENELSE{c}{f s_1}{f s_2}}} {\stateAbstr} \BSEval t o$}
			\dottedLine
			\RightLabel{\scriptsize by ind}
			\UnaryInfC{$\BSState {\IFTHENELSE{c}{s_1}{s_2}} \stateAbstr \BSEval t o$}
			\AxiomC{$o \neq \outOk \_$}
			\insertBetweenHyps{\hskip -12pt}
			\LeftLabel{\scriptsize $[\LOOP_{exit}]$}
			\BinaryInfC{$\BSState {\LOOP \left( \IFTHENELSE{c}{s_1}{s_2} \right)} \stateAbstr \BSEval t o$}
			\DisplayProof
		}
		&
		\makecell[c]{
			\hspace{-2mm}
			\vspace{-1cm}
			$\left\downarrow \vbox to 8mm{} \right.$
			\hspace{-1em}
			\AxiomC{$\istrue (\eval c \stateAbstr)$}
			\AxiomC{$\BSState{f s_1}{\stateAbstr} \BSEval t o \neq \outOk \_$}
			\UnaryInfC{$\BSState {\LOOP \left( f s_1 \right)} \stateAbstr \BSEval t o$}
			\insertBetweenHyps{\hskip -6pt}
			\BinaryInfC{$\BSState {\IFTHENELSE{c}{\LOOP (f s_1)}{\LOOP (f s_2)}} \stateAbstr \BSEval t o$}
			\DisplayProof
		}
	\end{tabular}
\end{table}

The assumption is shown in the lower left.
From there, we proceed upward by unfolding the big-step judgment.
At the dashed line, the induction hypothesis is used to pull $f$ into the statement without changing the outcome.
For the execution of the resulting if-statement, we assume w.~l.~o.~g.~that the condition $c$ evaluates to true; the other case is analogous.
Using the so-obtained facts, the desired statement is then established on the right, read from top to bottom.

We have proven this case just by trivial unfolding of the judgment and using induction.
The second case is a little more interesting: the loop completes one full iteration and is then recursively evaluated.
It is treated in the following diagram:

\begin{table}[H]
	\small
	\centering
	$\left\uparrow \vbox to 12mm{} \right.$
	\hspace{-1em}
	\AxiomC{$\istrue (\eval c \stateAbstr)$}
	\AxiomC{$\BSState{f s_1}{\stateAbstr} \BSEval {t_1} \outOk {\stateAbstr'}$}
	\RightLabel{\scriptsize wlog}
	\BinaryInfC{$\BSState {\IFTHENELSE{c}{f s_1}{f s_2}} {\stateAbstr} \BSEval {t_1} \outOk {\stateAbstr'}$}
	\dottedLine
	\RightLabel{\scriptsize by ind}
	\UnaryInfC{$\BSState {\IFTHENELSE{c}{s_1}{s_2}} \stateAbstr \BSEval {t_1} \outOk {\stateAbstr'}$}
	\AxiomC{$\istrue (\overbrace{\eval c {\stateAbstr'}}^{=\ \eval c \stateAbstr})$}
	\AxiomC{$\BSState {\LOOP \left( f s_1 \right)} {\stateAbstr'} \BSEval {t_2} o$}
	\RightLabel{\scriptsize $(*)$}
	\BinaryInfC{$\BSState {\IFTHENELSE{c}{\LOOP \left( f s_1 \right)}{\LOOP \left( f s_2 \right)}} {\stateAbstr'} \BSEval{t_2} o$}
	\dottedLine
	\RightLabel{\scriptsize by ind}
	\UnaryInfC{$\BSState {\LOOP \left( \IFTHENELSE{c}{s_1}{s_2} \right)} {\stateAbstr'} \BSEval {t_2} o$}
	\LeftLabel{\scriptsize $[\LOOP_{cont}]$}
	\BinaryInfC{$\BSState {\LOOP \left( \IFTHENELSE{c}{s_1}{s_2} \right)} {\stateAbstr} \BSEval {t_1 \traceConcat t_2} o$}
	\DisplayProof
\end{table}

Using induction twice we get two statements about execution of the transformed premises.
\cref{lemma-indep-spec} yields $\eval c {\stateAbstr'} = \eval c \stateAbstr$.
Hence, only the $\istrue{\!\!}$ case may have occured at $(*)$ (which was assumed w.~l.~o.~g.~in the upper left), giving us crucial information about an additional execution of $\LOOP(f s_1)$, rather than possibly $\LOOP(f s_2)$.
The desired conclusion then follows immediately from the obtained facts about $\eval c \stateAbstr$ and the executions of $f s_1$ and $\LOOP(f s_1)$.

This proves preservation for the terminating case and, in fact, also establishes the partial case for free.
Showing the diverging case is very similar in spirit using coinduction on $\BSEvalinf{\tau}$, which gives the result directly for most statements, and looking at the two cases $[\LOOP_{div}]$ and $[\LOOP'_{cont}]$ to handle special loops.
We do not show it here but remark that, to prove it in Rocq, an additional, well-chosen coinductive hypothesis is needed to make Rocq's syntactic guard checker happy, an endeavor that can prove quite tricky in some cases.
	
\begin{figure}
	\centering
		\begin{tikzpicture}[scale=1.7]
			\node at (0,1) {$S_1$};
			\node at (1,1) {$S'_1$};
			\node at (0,0) {$S_2$};
			\node at (1,0) {$S'_2$};
			\draw [->, shorten >=8pt,shorten <=7pt] (0,1) -- (1,1);
			\draw [->, shorten >=8pt,shorten <=7pt, dashed] (0,0) -- (1,0);
			\draw [shorten >=8pt,shorten <=8pt] (0,1) -- (0,0);
			\draw [shorten >=8pt,shorten <=8pt, dashed] (1,1) -- (1,0);
			\node [] at (-0.15,0.47) {\small $\equiv$};
			\node [] at (1.15,0.47) {\small $\equiv$};
			\node [] at (0.5,0.1) {\scriptsize $t$};
			\node [] at (0.5,1.1) {\scriptsize $t$};
			\node [] at (0.85,0.05) {\tiny $*$};
		\end{tikzpicture}

	\small
	\begin{centeredmath}
		\begin{gather*}
			\cdots\coord{0x}
			\SSStep{} \BSState {\LOOP \left( \IFTHENELSE{c}{s_1}\coord{A1}{s_2} \right)} k \coord{a1} \coord{Ax1}
			\SSStep{} \BSState {\IFTHENELSE{c}{\coord{Bx1} s_1}{s_2}} {\coord{B1} \Kseq {\LOOP (\cdots\hskip -0.0em)} k}
			\underset{\mathclap{eval\ c}}{\SSStep{}} \BSState {s_1} {\Kseq {\LOOP \coord{Cx} (\cdots\hskip -0.0em)} k}
			\SSStep{}^* \coord{d1} \BSState {\SKIP} {\Kseq {\LOOP \coord{Dx} (\cdots\hskip -0.0em)} k} \\[2em]
			\cdots
			\SSStep{} \BSState {\IFTHENELSE{c}{\LOOP (f s_1)}{\coord{A2} \LOOP (f s_2 \coord{Ax2})}} {f k}
			\underset{\mathclap{eval\ c}}{\SSStep{}} \BSState {\coord{Bx2} \LOOP \left( f \coord{B2} s_1 \right)} {f k} \coord{b2}
			\SSStep{} \BSState {\coord{Cx2} f s_1} {\Kseq {\LOOP(f s_1)} {f k}}
			\SSStep{}^* \coord{d2} \BSState {\SKIP} {\Kseq {\LOOP (f s_1)} {f k}}
		\end{gather*}
	\end{centeredmath}
	\begin{tikzpicture}[overlay,remember picture]
		\draw[->,shorten >=3pt,shorten <=3pt, out=140,in=40,distance=7mm]
		([yshift=2.2ex, xshift=1ex]d1) to ([yshift=2.2ex, xshift=-1.5ex]a1);
		\draw[->,shorten >=3pt,shorten <=3pt, out=-150,in=-30,distance=5mm]
		([yshift=-1ex, xshift=1ex]d2) to ([yshift=-1ex, xshift=-1.5ex]b2);
		
		\coordinate(Y1) at ([yshift=-1ex]A1);
		\coordinate(Y2) at ([yshift=2.3ex]A2);
		\draw[=,double,double distance=2pt] ([xshift=-2.6ex]0x|-Y1) to ([xshift=-2.6ex]0x|-Y2);
		\node at ([yshift=0.5ex,xshift=-1.4ex] 0x|-Y2) {$\scriptstyle f$};
		\draw[=,double,double distance=2pt] (A1|-Y1) to (A1|-Y2);
		\node at ([yshift=0.5ex,xshift=1.2ex] A1|-Y2) {$\scriptstyle f$};
		\draw[=,double,double distance=2pt] (Cx|-Y1) to (Cx|-Y2);
		\node at ([yshift=0.5ex,xshift=1.2ex] Cx|-Y2) {$\scriptstyle f$};
		\draw[=,double,double distance=2pt] (Dx|-Y1) to (Dx|-Y2);
		\node at ([yshift=0.5ex,xshift=1.2ex] Dx|-Y2) {$\scriptstyle f$};
		
		\draw[=,double,double distance=2pt,dashed,nfold,shorten >=2pt] (Ax1|-Y1) to (Bx2|-Y2);
		\node at ([yshift=1ex,xshift=0.2ex] Bx2|-Y2) {$\scriptscriptstyle 1$};
		\draw[=,double,double distance=2pt,dotted,nfold] (Bx1|-Y1) to (Ax2|-Y2);
		\node at ([yshift=1ex,xshift=-0.4ex] Ax2|-Y2) {$\scriptscriptstyle 0$};
		\draw[=,double,double distance=2pt,dashed,nfold,shorten >=2pt] (B1|-Y1) to (Cx2|-Y2);
		\node at ([yshift=1ex,xshift=0.2ex] Cx2|-Y2) {$\scriptscriptstyle 2$};
	\end{tikzpicture}
	\captionsetup{width=.95\linewidth}
	\caption{At the top, a general small-step simulation diagram.
	At the bottom, a small-step simulation diagram for loop unswitching.
	The original program execution is shown directly above the transformed execution; vertical double lines indicate which states need to match with each other.
	Environments and traces are omitted for brevity.}
	\label{fig-unswitching-simulation}
\end{figure}
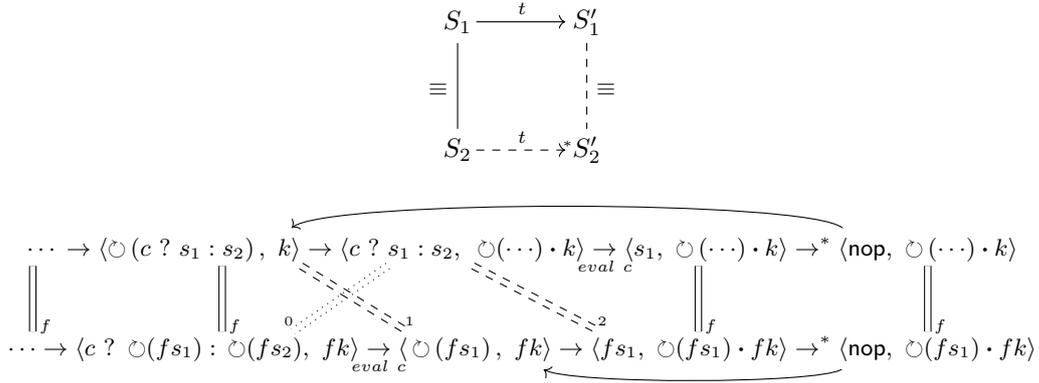

\paragraph*{Small-Step}

We now sketch the proof of $f$'s correctness in a small-step fashion.
Therefore a \emph{small-step simulation} together with an appropriate matching relation between source and target states must be established \cite{compcert-3}.
The nature of a small-step simulation is illustrated at the top of \cref{fig-unswitching-simulation}:
given two matching states $S_1 \equiv S_2$ and a step in the source program $S_1 \SSStep t S_2$, a state $S'_2$ has to be found such that $S_2 \SSStar t S'_2$ and $S'_1 \equiv S'_2$.
No strict one-to-one correspondence of steps is required: the target program may perform zero or multiple steps for a single step in the source program.

First, a match predicate $=_f$ between states of the original program $p$ and the transformed program $f p$ must be defined.
Outside of a special loop, this is simply $\SSState s k \stateAbstr =_f \SSState {f s} {f k} \stateAbstr$, where we lift $f$ inductively to continuations.
When encountering a special loop, things become interesting.
This process is illustrated at the bottom of \cref{fig-unswitching-simulation}:
Before entering the loop, execution is at the point shown on the left in the figure,
$\SSState {\LOOP\left(\IFTHENELSE{c}{s_1}{s_2}\right)} k \stateAbstr =_f \SSState {\IFTHENELSE{c}{\LOOP(f s_1)}{\LOOP(f s_2)}} {f k} \stateAbstr$.
In $p$, the next step first unfolds the loop and only afterwards, $c$ is evaluated.
In contrast, $c$ is evaluated immediately in $f p$.
Thus, $f p$ must wait until $p$ has taken two steps; the predicate $=_0$ is needed to allow matching the intermediate state.
Once $c$ has been evaluated in $p$, two steps can be performed at once in $f p$, and we return to the standard $=_f$ matching, which is maintained throughout the remainder of the loop.

If the loop executes normally (i.~e., it is not terminated by $\EXIT$ and does not diverge), execution in both $p$ and $f p$ reaches a $\SKIP$ statement, shown on the right in the figure.  
From there, both $p$ and $f p$ return to the beginning of their respective loops.  
However, these states do not match under $=_f$, so we additionally require $=_1$ and $=_2$ to maintain the matching before returning to the contents of the loops which again match via $=_f$.

To correctly consider special loops, the match definition must now also be lifted to the continuations, thereby carrying information about the evaluation of $c$ and its independence from the current statement and the relevant part of the continuations, in order to show that $\eval c \stateAbstr$ will always evaluate to true (or to false) in $p$.
During the execution of a special loop, a match between continuations could look as follows (where the $a_i$ are sub-statements of $s_1$):
\begin{gather*}
	a_n \kseq \dots \kseq a_1 \kseq \LOOP (\IFTHENELSE{c}{s_1}{s_2}) \kseq k
	\cong_\stateAbstr
	f a_n \kseq \dots \kseq f a_1 \kseq \LOOP (f s_1) \kseq f k, \\
	\emph{ given } \istrue (\eval c \stateAbstr) \wedge \indep c {s_1} \wedge \forall i, \indep c {a_i}.
\end{gather*}

It becomes significantly more fiddly to correctly consider \emph{nested} special loops, because the multiple different conditions are then independent over different parts of the continuation.

With the complete match definition, it can be shown that each step in $p$ corresponds to one or two steps in $f p$, or that two steps in $p$ correspond to one step in $f p$, and that the matching together with the $indep$ and $eval$ invariants are preserved.  
From this, the small-step simulation follows, which in turn establishes forward preservation $p \fwpreserv f p$.\\

A part of the complexity of this match definition comes from the fact that $\Cminor$ is defined inductively rather than via a control-flow graph.
Nevertheless, the match predicate must generally fulfill several functions simultaneously: in addition to matching the states, it must carry along all invariants that hold between $p$ and $f p$.
As a result, both the match predicate and the actual simulation proof often become bloated.
Notice that \cref{lemma-indep-spec}, the compact specification of $indep$, has not been used in the small-step proof.
Instead, in each part of the simulation proof -- that is, for each different step-match combination -- it must be shown that both independence and the result of $eval$ are preserved for that specific case.
This effectively breaks up and scatters \cref{lemma-indep-spec} throughout the entire simulation proof, reducing both separation of concerns and opportunities for proof (i.~e., code) reuse.
	\subsection{Loop Unrolling}

\definecolor{codebg}{rgb}{0.95,0.95,0.95}
\begin{figure}[H]
\centering
\begin{minipage}{0.35\textwidth}
\begin{minted}[bgcolor=codebg, tabsize=4]{c}
int x=1;
for (int i=1; i<11; i++)
	x = x*i;
return x;
\end{minted}
\end{minipage}
\begin{minipage}{0.08\textwidth}
\end{minipage}
\begin{minipage}{0.25\textwidth}
\begin{minted}[bgcolor=codebg]{c}
int x=1; int i=1;
x = x*i;
i++;
...
x = x*i;
i++;
return x;
\end{minted}
\end{minipage}
\begin{minipage}{0.08\textwidth}
\end{minipage}
\begin{minipage}{0.22\textwidth}
\begin{minted}[bgcolor=codebg]{c}
return 3628800;
\end{minted}
\end{minipage}
\caption{Left: loop that calculates $10!$. Middle: code after loop unrolling. Right: code after loop unrolling and CompCert's constant propagation pass.}
\label{fig-unrolling-example}
\end{figure}

\newcommand{\loopbody}{body_m}

\cref{fig-unrolling-example} showcases our verified loop unrolling pass in action, working together with constant propagation to statically calculate the value of $10!$.
In general, consider a loop of the form \textit{for (i=0; i<\textbf{m}; i++) do stmt} where \textbf{m} is a constant.
In \Cminor, the loop body may look like
$
\loopbody := \big(\SEQ{\SEQ{\big(\IFTHENELSE{i<^{\smash?}\textbf{m}}{\EXIT 0}{\SKIP}\big)}{stmt}}{(\STORE i {i+1})} \big).
$
The transformation takes statements of the form $\SEQ{(\STORE i 0)}{\BLOCK {\LOOP \loopbody}}$
and unrolls them into ${\SEQ{(\STORE i 0)}{\rep_m (\SEQ {stmt} {(\STORE i {i+1})})}}$
where $\rep_0 s := \SKIP$ and $\rep_{i+1} s := \SEQ s {\rep_i s}$.
As above, we define the transformation $f$ inductively over the statement structure, performing the loop unrolling when a matching statement is found and a few conditions are fulfilled, most notably that the condition $i <^? \mathbf{m}$ is independent of $stmt$ -- reusing the $indep$ predicate from before -- and that $stmt$ itself cannot exit out of the loop.
Bodies of an unrolled loop are not recursively unrolled themselves to avoid code explosion.

To show preservation in big-step style, we begin by introducing a notion $\BSLoop n t$. Used with an $\outExit{\!}{\!}\!$ or $\outPartial$ outcome, it means that the loop was exited or evaluation was truncated during the $n$'th iteration:

\begin{minipage}{\linewidth}
	\centering
	\rulefig{.25} {
		\BSState s \stateAbstr \BSEval t o
	}{
		\BSState s \stateAbstr \BSLoop 0 t o 
	}
	\rulefig{.4} {
		\BSState s \stateAbstr \BSEval {t_1} \outOk {\stateAbstr'} \hspace{2em} \BSState s {\stateAbstr'} \BSLoop i {t_2} o
	}{
		\BSState s \stateAbstr \BSLoop {i+1} {t_1 \traceConcat t_2} o 
	}
\end{minipage}


A similar notion $\BSLoopinf{n}{\tau}$ for divergence is introduced, meaning that the loop body diverged during its $n$'th iteration.
The proof of preservation then proceeds in four main steps:
\begin{enumerate}
	\item Show a few basic statements about a single execution of $\loopbody$ w.~r.~t.~the value of $\stateAbstr[i]$.
	\item Count the exact amount of performed loop iterations:
	Given $\stateAbstr[i] = 0$,
	\iftoggle{lipics}{\begin{alphaenumerate}}{\begin{itemize}}
		\item \makebox[\iftoggle{lipics}{5.6cm}{5.3cm}][l]{$\BSState {\LOOP \loopbody} \stateAbstr \BSEval t \outExit 0 \stateAbstr \Rightarrow$}
			$\BSState {\loopbody} \stateAbstr \BSLoop m t \outExit 0 \stateAbstr$
		\item $\BSState {\LOOP \loopbody} \stateAbstr \BSEval t \outPartial \iftoggle{lipics}{\hspace{0.8mm}}{} \Rightarrow  \exists l \leq m, \BSState {\loopbody} \stateAbstr \BSLoop {l} t \outPartial$
		\item \makebox[\iftoggle{lipics}{3.93cm}{3.7cm}][l]{$\BSState {\LOOP \loopbody} \stateAbstr \BSEvalinf \tau$}
			$\Rightarrow \exists l \leq m, \BSState {\loopbody} \stateAbstr \BSLoopinf {l} \tau$
	\iftoggle{lipics}{\end{alphaenumerate}}{\end{itemize}} 
	
	\item Unroll the loop as often as it was executed:
	\iftoggle{lipics}{\begin{alphaenumerate}}{\begin{itemize}}
		\item $\hspace{1.16cm} \BSState {\loopbody} \stateAbstr \BSLoop {m} t \outExit 0 \stateAbstr \Rightarrow \BSState {\rep_{m} {(\SEQ {stmt} {(\STORE i {i+1})}})} \stateAbstr \BSEval t \outOk \stateAbstr$
		\item $\forall l \leq m, \BSState {\loopbody} \stateAbstr \BSLoop {l} t \outPartial \hspace{0.4mm} \Rightarrow \BSState {\rep_l {(\SEQ {stmt} {(\STORE i {i+1})}})} \stateAbstr \BSEval t \outPartial$
		\item $\forall l \leq m, \BSState {\loopbody} \stateAbstr \BSLoopinf {l} \tau \hspace{1.22cm} \Rightarrow \BSState {\rep_l {(\SEQ {stmt} {(\STORE i {i+1})}})} \stateAbstr \BSEvalinf \tau$
	\iftoggle{lipics}{\end{alphaenumerate}}{\end{itemize}}

	\item Further unroll if necessary ($m$ times in total):
	\iftoggle{lipics}{\begin{alphaenumerate}}{\begin{itemize}}
		\item $\forall s, i \leq j, \BSState {\rep_i s} \stateAbstr \BSEval t \outPartial \Rightarrow \BSState {\rep_j s} \stateAbstr \BSEval t \outPartial$
		\item $\forall s, i \leq j, \BSState {\rep_i s} \stateAbstr \BSEvalinf \tau \hspace{1.33cm} \Rightarrow \BSState {\rep_j s} \stateAbstr \BSEvalinf \tau$
	\iftoggle{lipics}{\end{alphaenumerate}}{\end{itemize}}
\end{enumerate}

If the loop $\LOOP \loopbody$ terminates, it corresponds to exactly $m$ loop iterations (given $\stateAbstr[i]=0$).
This is not true when considering divergence or partial execution: there, we have \emph{at most} $m$ iterations before the loop body diverges or termination is truncated, respectively.
This requires us to prove several statements separately for the terminating, partial and diverging cases.
After bringing them in a slightly more general form, the statements of steps 2 and 3 can be proven inductively using reverse induction on the initial value of $\stateAbstr[i]$ together with the lemmas established in step 1.
Notice that a loop cannot terminate with an $\outOk{\!}$ outcome, therefore only $\outExit{\!}{\!}\!\!$ is relevant.
In step 1, \cref{lemma-indep-spec} is used as a key ingredient to reason about the value of $\stateAbstr[i]$.

Combined, we get the desired preservation:
any behavior of ${\SEQ {(\STORE i 0)} {\BLOCK {\LOOP \loopbody}}}$ induces the same behavior of the unrolled loop ${\SEQ {(\STORE i 0)} {\rep_m (\SEQ {stmt} {(\STORE i {i+1})})}}$.
For non-loop statements, preservation is straightforward due to the inductive definition --
we just have to be a little careful with the (co)induction hypotheses as we chose to not recursively unroll loops.

%
%
%
%

\subsection{While True}

A fun little example of the big-step approach is a transformation that eliminates the bodies of silent loops that never exit.
A statement is silent when it never emits any event, and it cannot exit if it does not contain an $\EXIT$ statement.
So, a loop $\LOOP s$ is transformed into $\LOOP \SKIP$, given $noexit\ s \wedge silent\ s$.
While it may be argued that this transformation is pointless, it is nonetheless correct: it refines the behavior of the code.

\begin{lemma}
  $noexit\ s \wedge \BSState {\LOOP s} \stateAbstr \BSEval t o \Rightarrow o = \outPartial$
  \label[lemma]{lem-noexit}
\end{lemma}

\begin{lemma}
  $silent\ s \Rightarrow \big( [\BSState s \stateAbstr \BSEval t o \Rightarrow t = \emptyTrace] \wedge [\BSState s \stateAbstr \BSEvalinf \tau \Rightarrow \tau = \emptyTrace] \big)$
  \label[lemma]{lem-silent}
\end{lemma}

With the above lemmas characterising $noexit$ and $silent$, the big-step preservation proof almost becomes trivial.
In particular, no attention has to be paid to the terminating case because of \cref{lem-noexit}.
For the partial and diverging case we use the fact that the trace is empty (\cref{lem-silent}).
The simple facts $s \BSEval \emptyTrace \outPartial$ and $\LOOP \SKIP \BSEvalinf \emptyTrace$ then finish the proof.



	\section{Discussion and Related Work}
\label{section-discussion}

We argue that both forms of semantics have their own strengths and should be used where they fit best.
While small-step semantics work well for many transformations, they are not as natural as big-step semantics for structural transformations like loop optimizations.
An inductively defined langauge such as \Cminor is a perfect fit for these kinds of transformations.
An advantage is that large parts of the correctness arguments follow immediately by structural induction.
Furthermore, big-step reasoning avoids the need to match individual reduction steps.
As a result, no complex match relation with explicitly threaded invariants over intermediate states is required.
Instead, proofs can often be decomposed into smaller, logical parts, increasing modularity and reusability.
For example, \cref{lemma-indep-spec} (indep\_spec) was used as an essential component for both loop unswitching and loop unrolling.

However, big-step proofs require duplicated effort, as the inductive and coinductive statements cannot be combined into a single one.
Also, while the partial case can sometimes be subsumed in the terminating case, sometimes it has to be treated separately.
Custom proof tactics can help reduce code duplication however -- our full definition and proof of loop unrolling only span about 700 lines of Rocq (excluding indep\_spec), despite having to show essentially three behavioral preservation statements.
Nevertheless, coinductive proofs for divergence are often technically subtle and difficult to construct correctly.

In contrast, small-step semantics naturally accommodate divergence and support uniform simulation-based arguments, but at the cost of more intricate stepwise reasoning and heavier invariant management, reducing the possibility for modularity and abstraction.
In particular, invariants (such as $indep$ for loop unswitching) often have to be threaded through the step relation and proven on-the-fly, which makes it harder to split a small-step proof into distinct, reusable parts.

Our behavioral framework requires semantics to be determinate to achieve full behavioral equivalence.
We accept this as most intermediate languages of CompCert's backend are determinate, including \Cminor and \texttt{RTL}.
Of course, we are not tied to \Cminor in our approach -- its major disadvantage is that it lacks support for a goto statement.
A better candidate could be \texttt{RTL}, which is not inductively defined but consists of a control-flow graph.
Its big-step semantics could abstract function calls and therefore enable reasoning about inter-procedural optimizations such as inlining.
A big-step loop construct could also be introduced which would then enable loop optimizations. \\


Often, loop transformations such as unrolling change the control-flow graph in a systematic way: there is a surjective mapping from the transformed CFG onto the original one that preserves the edges.
Gourdin, Monniaux and others \cite{cfg-morphisms-1, cfg-morphisms-2, cfg-morphisms-3} introduced a way to leverage this fact and verify such transformations on a small-step semantics similar to \texttt{RTL}.
Instead of directly showing a small-step simulation, they use translation validation \cite{translation-validation}: a powerful technique where witnesses are generated alongside the transformation that can be used to validate the correctness of a transformation a posteriori.
As a validator, they use an interpreter that performs a symbolic small-step execution, and they provide hints in order to verify the required matching invariants.
Building on this, they were able to verify several new complex transformations, including loop-independent code motion (based upon loop rotation) as well as more general lazy code motion strategies.

While these approaches give powerful results for a large class of transformations, they are more complicated than our ``native'' big-step approach.
As they require CFG morphisms, some transformations such as reordering of operations, loop fusion or loop unswitching do not natively fit into their framework but are, in our opinion, better suited for big-step proofs.
	
	\bibliography{bib}
	
\end{document}